\documentclass[11pt]{article}
\usepackage[a4paper,margin=25mm]{geometry}
\usepackage[T1]{fontenc}
\usepackage{lmodern,microtype}
\usepackage{amsmath,amssymb,amsthm,mathtools}
\usepackage{booktabs,tabularx,array}
\usepackage{algorithm,algpseudocode}
\usepackage{placeins,needspace}
\usepackage[round,authoryear]{natbib}
\usepackage{xcolor,hyperref}
\usepackage[capitalise,nameinlink,noabbrev]{cleveref}
\usepackage{aliascnt}
\hypersetup{colorlinks=true,linkcolor=blue!40!black,citecolor=blue!40!black,urlcolor=blue!40!black,pdftitle={Improved Approximations for Vehicle Routing with Nonuniform Speeds}}
\newtheorem{theorem}{Theorem}[section]
\newaliascnt{lemma}{theorem}
\newtheorem{lemma}[lemma]{Lemma}
\aliascntresetthe{lemma}
\crefname{lemma}{Lemma}{Lemmas}
\Crefname{lemma}{Lemma}{Lemmas}
\theoremstyle{definition}
\newaliascnt{remark}{theorem}

\aliascntresetthe{remark}
\crefname{remark}{Remark}{Remarks}
\Crefname{remark}{Remark}{Remarks}
\newcommand{\OPT}{\mathrm{OPT}}
\newcommand{\ALG}{\mathrm{ALG}}
\newcommand{\TSP}{\mathrm{TSP}}
\newcommand{\len}{\ell}
\DeclareMathOperator{\odd}{odd}
\title{Improved Approximations for Vehicle Routing\\with Nonuniform Speeds}
\author{%
Hong Li\\
\small School of Mathematics and Statistics, Yunnan University\\
\small \texttt{honglimath@126.com}
}
\date{}
\begin{document}

\maketitle

\begin{abstract}
We study vehicle routing with vehicles of different speeds on a complete undirected graph whose vertex set consists of a depot and a set of clients, where the distances satisfy the triangle inequality. Each vehicle has a specified speed, and if the total length traveled by a vehicle of speed $s$ is $L$, its completion time is $L/s$. In the heterogeneous traveling salesman problem (HetTSP), each vehicle executes one tour starting and ending at the depot, and the tours collectively visit all clients. The objective is to minimize the maximum completion time among the vehicles. We give a $6$-approximation algorithm for HetTSP, improving the previous $90(1+\delta)$-approximation for any fixed $\delta>0$.

We also consider two versions of the heterogeneous capacitated vehicle routing problem (HetCVRP). Each client has a demand, and the vehicles have identical capacities. A vehicle may execute several tours, each starting and ending at the depot, and reload at the depot between consecutive tours; the total demand delivered on each tour cannot exceed the vehicle capacity. In the split-delivery version of HetCVRP, the demand of a client may be divided among multiple visits, possibly by different vehicles. We give a $\frac92+2\sqrt3<7.965$-approximation algorithm for this problem. In the unsplit-delivery version of HetCVRP, the entire demand of each client must be delivered in a single visit. We give a $\frac{11}{2}+3\sqrt2<9.743$-approximation algorithm, improving the previous $450(1+\delta)$-approximation for any fixed $\delta>0$.
\end{abstract}

\section{Introduction}
\subsection{Model}\label{sec:model}

We study the following vehicle routing problems in the formulations used by \citet{GMNR}. Several vehicles, possibly with different speeds, operate from a common depot, and the objective is to minimize the maximum completion time among the vehicles. We first specify the graph, tours, vehicle routes, and makespan objective, and then define the uncapacitated problem and the two capacitated versions studied in this paper.

\paragraph{Graph, walks, tours, routes, and makespan.}
Let $G=(V,E)$ be a complete undirected graph with metric distances $d:V\times V\to\mathbb{R}_{\ge0}$. Thus $d(u,v)$ denotes the distance between vertices $u$ and $v$, and the distances satisfy the triangle inequality. A distinguished vertex $r\in V$ is the common depot, and $V^\circ=V\setminus\{r\}$ is the set of clients. There are $k$ vehicles with positive speeds $s_1,\ldots,s_k$.

A walk is a finite sequence of vertices, and its length is the sum of the distances between consecutive vertices, counting repeated traversals with multiplicity. A walk is closed if it starts and ends at the same vertex. We also allow the empty walk at a vertex, of length zero. A tour is either the empty walk at $r$ or a closed walk $r,v_1,\ldots,v_m,r$ with $m\ge1$, where $v_1,\ldots,v_m$ are distinct clients. Thus a tour containing a single client $v$ is $r,v,r$. A route of a vehicle is a finite sequence of tours executed by that vehicle. The length of a route is the sum of the lengths of its tours. If vehicle $i$ has a route of length $L_i$, its completion time is $L_i/s_i$. The objective in all problems considered below is to minimize the makespan
\[
\max_{i\in\{1,\ldots,k\}}\frac{L_i}{s_i}.
\]

\paragraph{Heterogeneous traveling salesman problem.}
In the heterogeneous traveling salesman problem (HetTSP), the route of each vehicle consists of a single tour, possibly the empty tour, and the tours of the vehicles must collectively visit every client in $V^\circ$. We denote the minimum possible makespan by $\OPT_{\mathrm{HetTSP}}$.

\paragraph{Capacitated setting.}

In the capacitated setting, every vehicle has the same capacity $Q>0$, and each client $v\in V^\circ$ has a demand $q_v\in(0,Q]$. The route of a vehicle may contain several tours, and the vehicle may reload at the depot between consecutive tours. We call a tour capacity-feasible if the total quantity delivered on it is at most $Q$. All client demands must be satisfied. We consider the following two versions of the heterogeneous capacitated vehicle routing problem (HetCVRP).

\paragraph{Split-delivery HetCVRP.}
In split-delivery HetCVRP, the demand $q_v$ of a client may be divided among several tours, possibly executed by different vehicles. We denote the minimum possible makespan by $\OPT_{\mathrm{SD}}$.

\paragraph{Unsplit-delivery HetCVRP.}
In unsplit-delivery HetCVRP, the entire demand $q_v$ of each client must be delivered in a single visit, and hence on a single tour by a single vehicle. We denote the minimum possible makespan by $\OPT_{\mathrm{UD}}$. Since every unsplit-delivery solution is also feasible for split-delivery HetCVRP, $\OPT_{\mathrm{SD}}\le \OPT_{\mathrm{UD}}$.

We impose $q_v\le Q$ in both versions. For unsplit-delivery HetCVRP, this condition is necessary for feasibility. For split-delivery HetCVRP, it restricts the problem to the same class of instances, although larger demands could otherwise be split among several tours.

\subsection{Related work}

The metric traveling salesman problem (TSP) asks for a minimum-length tour visiting all vertices when the distances satisfy the triangle inequality. A classical extension replaces the single salesman by several salesmen sharing a common depot and minimizes the length of the longest tour. \citet{FHK} studied several $k$-person versions of routing problems, including the $k$-traveling salesman problem ($k$-TSP), where $k$ denotes the number of salesmen, referred to as vehicles in this paper. For $k$-TSP, they proposed the tour-splitting approach: first compute one TSP tour and then divide it into $k$ parts, reconnecting each part to the depot. If TSP admits an $\alpha$-approximation, their analysis gives an $\alpha+1-1/k$ approximation for $k$-TSP, and hence in particular an $\alpha+1$ approximation. This is the equal-speed predecessor of HetTSP.

A second classical extension of TSP is obtained by imposing the capacitated setting described in \cref{sec:model}, leading to the capacitated vehicle routing problem (CVRP). In the single-depot setting relevant here, all vehicles have the same capacity, and the usual objective is to minimize the total length of the tours used to satisfy all client demands. The tour-splitting approach is also fundamental to approximation algorithms for CVRP: the classical algorithms of \citet{HRK,AG87,AG90} first compute a TSP tour and then divide its client order into capacity-feasible tours, reconnecting the resulting parts to the depot as needed. Given an $\alpha$-approximation algorithm for TSP, they obtain an $(\alpha+1)$-approximation for the unit-demand special case and the split-delivery version, and an $(\alpha+2)$-approximation for unsplit-delivery CVRP. For every $\alpha>1$, \citet{BTV} improved these guarantees to $\alpha+1-\varepsilon$ and $\alpha+2(1-\varepsilon)$, respectively, for a constant $\varepsilon>0$ depending on $\alpha$, with $\varepsilon>1/3000$ for $\alpha=3/2$. Taking $\alpha=3/2$, the ratio for unsplit-delivery CVRP was further improved to $3.194$ by \citet{FMR}, to below $3.176$ by \citet{ZX2026}, and to $3.159$ by \citet{EW2026}. \citet{ZX2025} studied the case where the capacity is a fixed constant. When a given number of vehicles operate in parallel, minimizing the maximum route length is also a natural objective; for equal-speed vehicles, this is equivalent to minimizing the makespan in this paper.

\citet{GMNR} studied vehicle routing with nonuniform speeds, including HetTSP and HetCVRP. Their main algorithmic contribution concerns HetTSP. They construct a Level-Prim tree that spans vertices in successive distance ranges from the depot, decompose it into subtours containing the depot, and assign these subtours to vehicles using ideas from scheduling on unrelated machines \citep{LST1990}. Their analysis gives a $90(1+\delta)$-approximation for HetTSP for any fixed $\delta>0$. Their HetCVRP result is obtained by transforming the capacitated instance into a HetTSP instance under a modified distance function. This transformation multiplies the HetTSP approximation factor by $5$, yielding an unsplit-delivery solution of makespan at most $450(1+\delta)\OPT_{\mathrm{SD}}$, and hence at most $450(1+\delta)\OPT_{\mathrm{UD}}$. Since this solution is also feasible for split-delivery HetCVRP, the same factor also gives a $450(1+\delta)$-approximation for that version.

Related approximation results have also been obtained for the heterogeneous Chinese postman problem \citep{LPLCWL2023}, parcel delivery using drones \citep{CGMR2024}, heterogeneous rooted tree/path cover problems \citep{PLYL2025}, and the min--max heterogeneous weighted delivery problem \citep{LYLX2026}.

\subsection{Our results}\label{sec:results}

We give approximation algorithms for HetTSP and the two versions of HetCVRP defined in \cref{sec:model}. For each algorithm, let $\ALG$ denote the makespan of its output solution.

\begin{theorem}[HetTSP]\label{thm:hettsp}
There is a polynomial-time algorithm for HetTSP that returns a feasible solution with $\ALG\le6\OPT_{\mathrm{HetTSP}}$.
\end{theorem}

\begin{theorem}[Split-delivery HetCVRP]\label{thm:sd}
There is a polynomial-time algorithm for split-delivery HetCVRP that returns a feasible solution with $\ALG\le(\frac92+2\sqrt3)\OPT_{\mathrm{SD}}$.
\end{theorem}

\begin{theorem}[Unsplit-delivery HetCVRP]\label{thm:ud}
There is a polynomial-time algorithm for unsplit-delivery HetCVRP that returns a feasible solution with $\ALG\le(\frac{11}{2}+3\sqrt2)\OPT_{\mathrm{UD}}$.
\end{theorem}

The factors in \cref{thm:sd,thm:ud} are less than $7.965$ and $9.743$, respectively. The unsplit-delivery algorithm satisfies the stronger guarantee
\[
\ALG\le\left(\frac{11}{2}+3\sqrt2\right)\OPT_{\mathrm{SD}}
\le\left(\frac{11}{2}+3\sqrt2\right)\OPT_{\mathrm{UD}}.
\]

\begin{table}[htb]
\centering
\caption{Comparison with the approximation guarantees of \citet{GMNR}. For unsplit-delivery HetCVRP, both guarantees also hold relative to $\OPT_{\mathrm{SD}}$.}
\label{tab:comparison}
\begin{tabularx}{\textwidth}{@{}>{\raggedright\arraybackslash}X>{\raggedright\arraybackslash}X>{\raggedright\arraybackslash}X@{}}
\toprule
Problem & Previous guarantee & This paper\\
\midrule
HetTSP & $90(1+\delta)$ & $6$\\
Split-delivery HetCVRP & $450(1+\delta)$ (implied) & $\frac92+2\sqrt3<7.965$\\
Unsplit-delivery HetCVRP & $450(1+\delta)$ & $\frac{11}{2}+3\sqrt2<9.743$\\
\bottomrule
\end{tabularx}
\end{table}

The improvement in the HetTSP guarantee cannot be explained by optimizing the free parameter in the displayed analysis of \citet{GMNR}. Although \citet{GMNR} note that their constants were not optimized, the best choice of this parameter reduces the coefficient multiplying $1+\delta$ only from $90$ to $4(1+\sqrt[3]{6})^3\approx89.4286$.\footnote{Combining Lemmas~2 and~3 of \citet{GMNR}, the coefficient before the factor $1+\delta$ is $F(\mu)=(1+\mu)(4+4\mu)+(6+6/\mu)\max\{6,4+4/\mu\}$ for $\mu\ge1$. Its minimum is attained at $\mu=\sqrt[3]{6}$ and equals $4(1+\sqrt[3]{6})^3$.}

The proofs, including the stronger bound for unsplit-delivery HetCVRP, are given in \cref{sec:main-proofs} after the common analysis of the algorithms.

\subsection{Techniques}

\paragraph{Algorithmic framework.}
Our algorithms for HetTSP and HetCVRP use the same basic framework. We order the vehicles from fastest to slowest and consider a tentative value for the makespan. For a client, we call the walk from the depot to the client and back to the depot its round trip. We then identify the slowest vehicle that can complete this round trip within the tentative makespan, which we call the client's slowest reachable vehicle. Every slower vehicle would require more than the tentative makespan for this round trip. We then divide the ordered vehicles into consecutive groups. For each group, the clients whose slowest reachable vehicle belongs to that group form a layer client set. The layer client sets are pairwise disjoint and form a partition of the clients, ordered from those that require faster vehicles to those that can also be served by slower vehicles.

For each layer client set, we construct a tour through its clients and the depot using the Christofides--Serdjukov algorithm \citep{Christofides,Serdjukov}. We compute a minimum spanning tree on the layer client set together with the depot using the original distances, add a minimum-length perfect matching on its odd-degree vertices, and shortcut an Euler traversal to obtain the tour. For HetCVRP, we also apply the tour-splitting approach to obtain capacity-feasible tours. Thus each layer client set has a closed walk starting and ending at the depot: one tour for HetTSP, or a concatenation of capacity-feasible tours for HetCVRP. We concatenate these closed walks in their order and assign subwalks of the resulting walk to the vehicles from fastest to slowest, proceeding in the order of the walk. Each assigned subwalk is extended as far as possible, provided that adding direct connections from its endpoints to the depot keeps the route within the vehicle's permitted length. This permitted length is its speed times a makespan bound computed from the tentative makespan, the constructed closed walks, and the vehicle groups. This bound may exceed the tentative makespan.

The tentative makespan guides the construction but is not assumed to be feasible. The algorithm enumerates a polynomial number of tentative makespans and partitions into vehicle groups and returns the best resulting solution.

\paragraph{Why the framework improves the guarantees.}
With nonuniform speeds, it is not enough to divide a tour into short parts: a short part may still be too far from the depot for a slow vehicle to serve. The approach of \citet{GMNR} handles this interaction by constructing a Level-Prim tree, decomposing it into subtours containing the depot that satisfy conditions for assignment, and then assigning these subtours through a scheduling argument. Our approach differs in how the clients are organized, how the constructed walks are assigned, and how capacity is incorporated.

The Level-Prim tree of \citet{GMNR} organizes vertices by distance ranges, with compatibility with vehicle speeds enforced through the subsequent decomposition and assignment. Our layer client sets use the slowest reachable vehicle of each client directly. Every client in a layer client set can be reached within the tentative makespan by the first, and hence a fastest, vehicle of its vehicle group, while no vehicle appearing after that group in the speed order can make the same round trip in that time. This gives an upper bound on distances for the final assignment and a lower bound for the tour analysis. For the analysis, we fix an optimal HetTSP solution for HetTSP and an optimal split-delivery HetCVRP solution for both versions of HetCVRP. Among the enumerated tentative makespans, one gives every client the same slowest reachable vehicle as the makespan of this optimal solution. For that value, the clients of a layer client set are served only by vehicles in the corresponding vehicle group or in earlier, faster groups. The routes of the optimal solution can then be shortcut to tours whose client sets are pairwise disjoint and cover the layer client set. Deleting one edge incident to the depot from each nonempty tour produces a spanning tree, and each deleted edge has length at least the common lower bound on the clients' distances from the depot. The minimum spanning tree chosen by our algorithm has no greater length, so its length bound also benefits from these deletions. Adding a minimum-length perfect matching and shortcutting then gives an upper bound of three halves of the available route length minus the depot saving obtained by combining these deletions with the bounds on the optimal vehicle routes. Using only the approximation guarantee of an $\alpha$-approximation algorithm for TSP does not capture this saving, even when using the better-than-$3/2$ metric-TSP algorithm of \citet{KKO}.

The depot saving becomes useful when we concatenate the closed walks for different layer client sets and divide the resulting walk among the vehicles. In this allocation, a vehicle may be assigned a subwalk that spans consecutive layer client sets, while the returns to the depot already present in the closed walks are retained. The lengths of the assigned subwalks are analyzed together with the direct connections from their endpoints to the depot. The vehicle groups are chosen from the slowest end toward the fastest so that the depot savings in the tour bounds offset part of the length of these direct connections. In this way, tour construction and vehicle assignment are analyzed together, without the separate scheduling step used by \citet{GMNR}.

Capacity is handled directly within the same framework. After building the tour for a layer client set, we divide it into capacity-feasible tours using the tour-splitting approach. The additional length is bounded by an additive value depending on the client demands and their distances from the depot. Thus capacity enters the same analysis as an additive term, instead of through the reduction of \citet{GMNR} to HetTSP under a modified metric, which multiplies the approximation factor by $5$.

For HetTSP, the factor $6$ follows by bounding two contributions separately: the tours for earlier layer client sets together with the direct connections introduced by allocation, and the tour for the current layer client set after its depot saving. The vehicle groups are chosen so that each contribution is at most three times the total length available to the vehicles up to and including the first vehicle of the current group. For HetCVRP, the additional length from tour splitting changes this balance, and the vehicle groups are chosen accordingly to obtain the two capacitated approximation factors.

\subsection{Organization}
\Cref{sec:preliminaries} introduces the notation and recalls results from previous work on the tour-splitting approach. \Cref{sec:algorithm} presents the complete algorithm. \Cref{sec:fixed-groups} shows that, for any fixed partition into vehicle groups, the construction produces feasible vehicle routes and establishes the length bounds needed for the approximation analysis. \Cref{sec:analysis-groups} explains how the partition into vehicle groups is chosen and derives the approximation guarantees for HetTSP and the two versions of HetCVRP. \Cref{sec:running-time} analyzes the running time. Appendix~\ref{app:tour-splitting} recalls the classical tour-splitting constructions and records the additional details used by the algorithm and the running-time analysis.

\section{Preliminaries}\label{sec:preliminaries}

\subsection{Notation}
Throughout, $d(u,v)$ denotes the distance between vertices $u$ and $v$, and the edge $uv$ has length $d(u,v)$. For a multigraph or collection of tours $H$, we also write $\len(H)$ for its total length, counting edge occurrences with multiplicity. Concatenating the tours of a route at $r$ gives a closed walk of the same length. For $X\subseteq V^\circ$, let $\TSP(X)$ be the minimum length of a tour starting and ending at $r$ and visiting exactly the clients in $X$, with $\TSP(\varnothing)=0$.

Deleting visits from a walk and replacing the intervening subwalks by direct edges is called shortcutting; by the triangle inequality, it never increases length. In particular, a closed walk starting and ending at $r$ can be shortcut to a tour that starts and ends at $r$ and visits each client appearing in the walk exactly once. When constructing a capacitated solution, we retain the visits to the depot used for reloading and apply shortcutting only within the subwalk between two consecutive visits to the depot. If a client is visited more than once within such a subwalk, its deliveries are combined at one retained visit, preserving the total quantity delivered and capacity feasibility. The retained depot visits therefore divide the resulting closed walk into a sequence of capacity-feasible tours, which forms the vehicle's route.

Assume $d(r,v)>0$ for each client. A client at zero distance from the depot in a pseudometric can be served for free before the construction. If the client set is empty, return a route consisting of the empty tour for each vehicle. We henceforth let $n=|V^\circ|>0$.

We index the vehicles from fastest to slowest, breaking ties according to a fixed order. Vehicle $i$ has speed $s_i$. For later use, define
\begin{equation}\label{eq:speeds}
 s_1\ge\cdots\ge s_k>0,\qquad
 S_0=0,\quad S_i=\sum_{j=1}^i s_j,\quad s_{k+1}=0.
\end{equation}

For an undirected multigraph $H$, let $\odd(H)$ denote its vertices of odd degree. We use $\uplus$ for the union of edge multisets, retaining multiplicities.

\subsection{Classical tour splitting}\label{sec:partition}

For an instance of split-delivery HetCVRP or unsplit-delivery HetCVRP and a client set $X\subseteq V^\circ$, define
\[
R_Q(X)=\frac2Q\sum_{v\in X}q_vd(r,v).
\]

We use the classical tour-splitting results of \citet{HRK,AG87}. In the notation of this paper, we use the following form.

\begin{lemma}[\citealp{HRK,AG87}]\label{lem:partition}
Given a tour $T$ starting and ending at $r$ and visiting exactly the clients in $X$, one can construct a sequence of capacity-feasible tours delivering all demands in $X$ with total length at most
\begin{equation}\label{eq:partition-bound}
 \len(T)+\lambda R_Q(X),
\end{equation}
where $\lambda=1$ for the split-delivery version and $\lambda=2$ for the unsplit-delivery version.
\end{lemma}

The construction partitions the client order of $T$ into capacity-feasible tours. In the unsplit-delivery version, each client is contained in exactly one of the resulting tours and its entire demand is delivered there. In the split-delivery version, the demand of a client may be split between two consecutive resulting tours, so that the client is visited in both tours and the quantities delivered at these visits sum to its demand.

Appendix~\ref{app:tour-splitting} recalls the classical constructions underlying \cref{lem:partition}. We follow the presentation of the same classical constructions in \citet[Theorem~5]{BTV} and record the additional details used in this paper, including those needed for the running-time analysis.

\section{The algorithm}\label{sec:algorithm}
We first describe the construction for a tentative makespan $\tau>0$ and any partition into vehicle groups. We specify below which tentative makespans are retained and which finitely many partitions into vehicle groups are considered by the complete algorithm. In \cref{lem:partition}, $\lambda$ is the coefficient of the additional length introduced by tour splitting, with $\lambda=1$ for split-delivery HetCVRP and $\lambda=2$ for unsplit-delivery HetCVRP. To present the three algorithms uniformly, we also set $\lambda=0$ for HetTSP, where no tour-splitting term is present.

\subsection{Slowest reachable vehicles and layer client sets}\label{sec:reachability}
For a time value $\omega>0$, define
\begin{equation}\label{eq:eligibility}
 e_\omega(v)=\max\{i\in\{1,\ldots,k\}:2d(r,v)\le \omega s_i\},
\end{equation}
with value zero if no vehicle can make this round trip. When $e_\omega(v)>0$, vehicle $e_\omega(v)$ is the last vehicle in the fixed order that can travel from the depot to $v$ and return within time $\omega$. Since the vehicles are ordered by speed, it is a slowest such vehicle; we call it the slowest reachable vehicle of $v$ at $\omega$. A tentative makespan $\tau$ is retained if and only if every client has a slowest reachable vehicle at $\tau$.

Write such a partition as
\begin{equation}\label{eq:vehicle-groups}
 [a_1,b_1],\ldots,[a_m,b_m],\qquad a_1=1,\quad a_{t+1}=b_t+1,\quad b_m=k.
\end{equation}
The vehicle group $[a_t,b_t]$ consists of the vehicles with indices $a_t,a_t+1,\ldots,b_t$. We call vehicles $a_t$ and $b_t$ the first and the last vehicle of the group; they are a fastest and a slowest vehicle of the group, respectively. For this group, let $A_t$ consist of the clients whose slowest reachable vehicle at $\tau$ belongs to it. Equivalently,
\begin{equation}\label{eq:layers}
 A_t=\{v\in V^\circ:a_t\le e_{\tau}(v)\le b_t\}.
\end{equation}
Some layer client sets may be empty; the sets $A_1,\ldots,A_m$ are pairwise disjoint and their union is $V^\circ$. Recall that $s_{k+1}=0$. By the definition of the slowest reachable vehicle, every $v\in A_t$ satisfies
\begin{equation}\label{eq:layer-radius}
 \frac{\tau s_{b_t+1}}2<d(r,v)\le\frac{\tau s_{a_t}}2.
\end{equation}
A layer client set need not be served by the vehicles in its corresponding vehicle group.

\subsection{Constructing a tour for a layer client set}\label{sec:tours}
If $A_t$ is empty, let $T_t$ be the empty tour. Otherwise, compute a minimum spanning tree $F_t$ on the complete graph with vertex set $A_t\cup\{r\}$ using the original distances $d$. Compute a minimum-length perfect matching $J_t$ on $\odd(F_t)$ with respect to $d$, find an Euler traversal of $F_t\uplus J_t$ starting at $r$, and shortcut it to a tour $T_t$ through $A_t$ and the depot. This is the Christofides--Serdjukov algorithm \citep{Christofides,Serdjukov} applied to $A_t\cup\{r\}$.

\subsection{Concatenation and ordered allocation}\label{sec:allocation}
For HetTSP, put $C_t=T_t$. For HetCVRP, apply \cref{lem:partition} to $T_t$ and let $C_t$ be the closed walk obtained by concatenating the resulting capacity-feasible tours at $r$ in their returned order, retaining all depot visits and delivered quantities. It starts and ends at the depot and serves exactly the clients of $A_t$. If $A_t$ is empty, use the empty closed walk. Concatenate the closed walks at the depot as $W=C_1\circ\cdots\circ C_m$.

A service event corresponds to a client visit in $W$ and, in HetCVRP, records the quantity delivered at that visit. For HetTSP and unsplit-delivery HetCVRP, each client has exactly one service event. For split-delivery HetCVRP, a client may have two service events, corresponding to its visits in the capacity-feasible tours produced by \cref{lem:partition}; the associated delivered quantities sum to its demand. The allocation never divides a service event. The service events are ordered according to their occurrences along $W$. For an event $z$, let $v(z)$ be its client, and let $p(z)$ be the length along $W$ from its start to the occurrence of $z$.

Define
\begin{equation}\label{eq:P-Gamma}
 P_t=\sum_{j=1}^t\len(C_j),\qquad \Gamma_1=0,\quad \Gamma_t=\sum_{j=2}^t(a_j-a_{j-1})s_{a_j}.
\end{equation}
Thus $P_t$ is the total length of the closed walks for the first $t$ layer client sets. The role of $\Gamma_t$ in accounting for the direct connections to the depot is established in \cref{lem:allocation}. For the current tentative makespan $\tau$ and partition into vehicle groups, define
\begin{equation}\label{eq:allocation-makespan}
 M=\max\left\{\tau,\ \max_{1\le t\le m}\frac{P_t+\tau\Gamma_t}{S_{a_t}}\right\}.
\end{equation}
This is the makespan bound used for allocation and need not equal the tentative makespan $\tau$.

Process the vehicles in the order $1,\ldots,k$, that is, from the fastest to the slowest. When vehicle $j$ is considered, it begins with the first unassigned service event $z$ in the order along $W$ and chooses the last event $z'$ in this order satisfying
\begin{equation}\label{eq:segment}
 d(r,v(z))+p(z')-p(z)+d(v(z'),r)\le Ms_j.
\end{equation}
All service events from $z$ through $z'$ are assigned to vehicle $j$. The subwalk of $W$ between their occurrences is its assigned subwalk. Adding a direct connection from the depot to $v(z)$ and another from $v(z')$ back to the depot gives a closed walk whose length is bounded by \cref{eq:segment}. For HetTSP, this closed walk is shortcut to a single tour starting and ending at $r$ and visiting exactly the assigned clients; this tour forms the vehicle's route. For both versions of HetCVRP, all intermediate visits to the depot are retained, and each subwalk between two consecutive visits to the depot is a tour. These tours, in their order along the closed walk, form the vehicle's route.

For all three problems, the vehicle's route has no greater length than this closed walk. For fixed $z$, the length of the closed walk is nondecreasing as $z'$ moves forward in the order of the service events, by the triangle inequality. Thus the last event $z'$ satisfying \cref{eq:segment} can be found by scanning.

Each unused vehicle receives the empty tour. \Cref{lem:allocation} proves that all service events are assigned, that the resulting routes form a feasible solution to the corresponding problem, and that this solution has makespan at most $M$.

\subsection{Enumerating tentative makespans and vehicle groups}\label{sec:vehicle-groups}

The tentative makespans are the distinct values in
\begin{equation}\label{eq:critical-times}
 \Theta=\left\{\frac{2d(r,v)}{s_i}:v\in V^\circ,\ i\in\{1,\ldots,k\}\right\}.
\end{equation}
There are at most $nk$ such values. The largest of them is retained: it is at least $2d(r,v)/s_1$ for every client $v$, so every client has a slowest reachable vehicle there. Hence \cref{alg:main} always constructs at least one solution.

The parameter $K$ is used to determine the partition into vehicle groups. We set $K=K_\lambda$, where
\begin{equation}\label{eq:K-values}
 K_0=3,\qquad K_1=2+\sqrt3,\qquad K_2=2+\frac3{\sqrt2}.
\end{equation}
Thus $K=3$ for HetTSP, $K=2+\sqrt3$ for split-delivery HetCVRP, and $K=2+3/\sqrt2$ for unsplit-delivery HetCVRP. These choices are derived in \cref{sec:factors}.

For a fixed value of $K$, an auxiliary parameter $\eta\in[0,1]$ determines a partition into vehicle groups. We first describe the partition for a given $\eta$, and then show that only finitely many values of $\eta$ need to be considered.

The groups are constructed from the slowest vehicles toward the fastest. During this construction, denote the groups in the order in which they are generated by
\[
[\bar a_1,\bar b_1],[\bar a_2,\bar b_2],\ldots.
\]
Set $\bar b_1=k$. Once $\bar b_t$ is fixed, define
\begin{equation}\label{eq:partition-rule}
 \bar a_t=\min\left\{j\in\{1,\ldots,\bar b_t\}:
 \left(K-\frac32\right)S_{\bar b_t}
 +\eta\left(j-\frac{\bar b_t}{2}\right)s_{\bar b_t+1}
 \le KS_j\right\}.
\end{equation}
If $\bar a_t=1$, stop and set $m=t$. Otherwise, set $\bar b_{t+1}=\bar a_t-1$ and continue. The minimum in \cref{eq:partition-rule} is always defined. Indeed, $j=\bar b_t$ satisfies the defining inequality because $\left(K-\frac32\right)S_{\bar b_t}
+\eta\frac{\bar b_t}{2}s_{\bar b_t+1}
\le (K-1)S_{\bar b_t}
<KS_{\bar b_t}$, where $\bar b_t s_{\bar b_t+1}\le S_{\bar b_t}$ and $\eta\le1$.

After the construction stops, reverse the numbering of the groups so that it agrees with the fixed ordering of the vehicles from fastest to slowest. Namely, for $t=1,\ldots,m$, set
\[
a_t=\bar a_{m+1-t},
\qquad
b_t=\bar b_{m+1-t}.
\]
This gives the partition
\[
[a_1,b_1],\ldots,[a_m,b_m]
\]
in \cref{eq:vehicle-groups}. The rule in \cref{eq:partition-rule} is chosen so that, in the analysis, the depot saving in the tour bound for each layer client set offsets part of the contribution of its vehicle group to the direct connections; see \cref{sec:accumulated-analysis}.

Although $\eta$ ranges over the interval $[0,1]$, only finitely many values need to be considered. For indices $1\le j\le b\le k$, consider the inequality in \cref{eq:partition-rule} with $\bar b_t=b$ and candidate index $j$. It is affine in $\eta$, so whether it holds can change only where equality occurs. If $(2j-b)s_{b+1}\ne0$, this equality occurs only at
\begin{equation}\label{eq:breakpoints}
 \eta=\frac{2KS_j-(2K-3)S_b}{(2j-b)s_{b+1}},
 \qquad
 1\le j\le b\le k.
\end{equation}
Retain the values in \cref{eq:breakpoints} that lie in $[0,1]$, and call them breakpoints. If $(2j-b)s_{b+1}=0$, the inequality is independent of $\eta$ and therefore creates no breakpoint.

Consider two values of $\eta$ in the same open interval between consecutive breakpoints. Since $\bar b_1=k$ is fixed, whether each inequality in the definition of $\bar a_1$ holds is the same for the two values of $\eta$, and hence \cref{eq:partition-rule} gives the same $\bar a_1$. Therefore $\bar b_2=\bar a_1-1$ is also the same. Repeating this argument shows that all subsequently generated groups and the stopping value $m$ are the same. Hence the resulting partition into vehicle groups is unchanged throughout the interval.

Let $\mathcal E_K$ contain all breakpoints in $[0,1]$, the endpoints $0$ and $1$, and the midpoint of each interval between two consecutive distinct values among them. For every $\eta\in\mathcal E_K$, perform the construction above and let $\mathcal G_K$ be the family of the resulting partitions into vehicle groups. Since there are $O(k^2)$ breakpoints, both $\mathcal E_K$ and $\mathcal G_K$ have size $O(k^2)$.

For the corresponding value of $K$, the algorithm considers every $\tau\in\Theta$ that is retained by \cref{sec:reachability} together with every partition in $\mathcal G_K$.

\subsection{The complete algorithm}
\Cref{alg:main} presents the three algorithms together using the parameter $\lambda$ introduced at the beginning of \cref{sec:algorithm} and $K=K_\lambda$ from \cref{eq:K-values}. For every retained tentative makespan $\tau\in\Theta$ and every partition of the ordered vehicles into consecutive vehicle groups in $\mathcal G_K$, it constructs one route for each vehicle. Among the resulting solutions, it returns one with the smallest makespan $\ALG$.

\algrenewcommand\algorithmicrequire{\textbf{Input:}}
\algrenewcommand\algorithmicensure{\textbf{Output:}}
\begin{algorithm}[!ht]\small
\caption{Routing with nonuniform speeds}\label{alg:main}
\begin{algorithmic}[1]
\Require Metric distances $d$ on $V$, depot $r$, vehicle speeds $(s_i)$, and $\lambda\in\{0,1,2\}$; demands $(q_v)$ and capacity $Q$ if $\lambda>0$.
\Ensure One route for each vehicle, forming a feasible HetTSP solution if $\lambda=0$, a feasible split-delivery HetCVRP solution if $\lambda=1$, or a feasible unsplit-delivery HetCVRP solution if $\lambda=2$.
\If{$V^\circ=\varnothing$}
 \State \Return the solution in which every vehicle receives the empty tour.
\EndIf
\State Index the vehicles from fastest to slowest as in \cref{eq:speeds}, compute $S_i$, set $s_{k+1}=0$, and use $K=K_\lambda$ from \cref{eq:K-values}.
\State Form $\Theta$ by \cref{eq:critical-times} and $\mathcal E_K$ as in \cref{sec:vehicle-groups}.
\State For each $\eta\in\mathcal E_K$, construct a partition of the ordered vehicles into consecutive vehicle groups as in \cref{sec:vehicle-groups}; collect the resulting partitions in $\mathcal G_K$.
\State Initialize the best solution as unset.
\For{each tentative makespan $\tau\in\Theta$}
 \State Compute $e_{\tau}(v)$, the index of the slowest reachable vehicle, for every client $v$.
 \If{$e_{\tau}(v)=0$ for some client $v$}
  \State Skip the current $\tau$ and continue with the next tentative makespan in $\Theta$.
 \EndIf
 \For{each partition into vehicle groups $([a_t,b_t])_{t=1}^m\in\mathcal G_K$}
  \For{$t=1,\ldots,m$}
   \State Let $A_t$ be the layer client set associated with the vehicle group $[a_t,b_t]$, as defined in \cref{eq:layers}.
   \State Construct the tour $T_t$ through the clients in $A_t$ and the depot as in \cref{sec:tours}.
   \State Construct the closed walk $C_t$ associated with $A_t$ and record its service events as in \cref{sec:allocation}.
  \EndFor
  \State Concatenate $C_1,\ldots,C_m$ at the depot to obtain $W$, and compute $P_t$, $\Gamma_t$, and $M$ by \cref{eq:P-Gamma,eq:allocation-makespan}.
  \State Assign consecutive service events of $W$ to the vehicles using \cref{eq:segment}, and construct one route for each vehicle as in \cref{sec:allocation}.
  \State Keep the resulting solution if no solution has yet been kept or if its makespan is smaller.
 \EndFor
\EndFor
\State \Return the best solution.
\end{algorithmic}
\end{algorithm}
\FloatBarrier

\section{Feasibility and tour bounds}\label{sec:fixed-groups}
For every retained tentative makespan and every partition of the ordered vehicles into consecutive vehicle groups as in \cref{eq:vehicle-groups}, the construction produces a set of vehicle routes; we call the resulting solution a constructed solution. We first show that every constructed solution is feasible.

\begin{lemma}[Feasibility of the ordered allocation]\label{lem:allocation}
With $M$ defined by \cref{eq:allocation-makespan}, every service event for a client in $A_t$ is assigned to one of the vehicles $1,\ldots,a_t$, for every $t=1,\ldots,m$. Hence the procedure in \cref{sec:allocation} assigns all service events and constructs one route for each vehicle. These routes form the constructed solution, which is feasible for HetTSP if $\lambda=0$, split-delivery HetCVRP if $\lambda=1$, or unsplit-delivery HetCVRP if $\lambda=2$, and has makespan at most $M$. For both versions of HetCVRP, every tour in every vehicle's route is capacity-feasible.
\end{lemma}
\begin{proof}
By \cref{eq:allocation-makespan}, $M\ge\tau$ and
\begin{equation}\label{eq:allocation-condition}
 P_t+\tau\Gamma_t\le MS_{a_t}\qquad(t=1,\ldots,m).
\end{equation}

Suppose, for a contradiction, that some service event for a client in $A_t$ remains unassigned after vehicle $a_t$ has been processed, and choose the smallest such $t$. Put $i=a_t$, and set $a_0=0$ for this proof. Since an event remains unassigned after vehicle $i$ has been processed, an unassigned event is present whenever any vehicle $j\le i$ is considered.

We first show that every vehicle $j=1,\ldots,i$ can take at least its first unassigned service event. Let $h\in\{1,\ldots,t\}$ be the unique index such that
\[
 a_{h-1}<j\le a_h.
\]
By the minimality of $t$, for every $\ell<h$ all service events for clients in $A_\ell$ have been assigned by the time vehicle $a_\ell$ has been processed. Since $a_\ell<j$, none of these events remains when vehicle $j$ is considered. Hence the first unassigned event at that time belongs to some layer client set $A_q$ with $q\ge h$. If $v$ is its client, then the upper bound in \cref{eq:layer-radius} and the ordering of the vehicle speeds give
\[
 2d(r,v)\le\tau s_{a_q}\le\tau s_{a_h}\le\tau s_j\le Ms_j.
\]
Thus \cref{eq:segment} is satisfied when this event is taken alone, so vehicle $j$ is assigned at least one service event.

For $1\le j\le i$, let $z_j$ be the first service event assigned to vehicle $j$, and let $z_{i+1}$ be the first event still unassigned after vehicle $i$ has been processed. The preceding argument shows that all these events exist. Since the allocation assigns consecutive service events and extends the block assigned to vehicle $j$ as far as \cref{eq:segment} permits, $z_{j+1}$ is the first event following that block, and vehicle $j$ cannot also take it. Therefore
\[
 d(r,v(z_j))+p(z_{j+1})-p(z_j)+d(v(z_{j+1}),r)>Ms_j
 \qquad(j=1,\ldots,i).
\]
Summing these inequalities over $j=1,\ldots,i$ telescopes the terms involving $p$. Since the subwalk of $W$ from its start to $z_1$ has length at least $d(r,v(z_1))$, we have $p(z_1)\ge d(r,v(z_1))$, and hence
\begin{equation}\label{eq:failed-extension}
 p(z_{i+1})+d(v(z_{i+1}),r)
 >MS_i-2\sum_{j=2}^{i}d(r,v(z_j)).
\end{equation}

For $a_{h-1}<j\le a_h$, the argument above shows that the client $v(z_j)$ belongs to $A_h$ or a later layer client set. Hence the upper bound in \cref{eq:layer-radius} gives
\[
 2d(r,v(z_j))\le\tau s_{a_h}.
\]
Since $i=a_t$, summing over the corresponding ranges of vehicle indices yields
\[
 2\sum_{j=2}^{a_t}d(r,v(z_j))
 \le
 \tau\sum_{h=2}^t(a_h-a_{h-1})s_{a_h}
 =\tau\Gamma_t.
\]
Moreover, service events are assigned in their order along $W$, and by assumption an event for a client in $A_t$ remains unassigned after vehicle $a_t$ has been processed. Therefore $z_{i+1}$ occurs in $C_1\circ\cdots\circ C_t$. The subwalk of $W$ from this occurrence to the final visit to the depot in $C_t$ has length at least $d(v(z_{i+1}),r)$, and consequently
\[
 p(z_{i+1})+d(v(z_{i+1}),r)\le P_t.
\]
Combining this inequality with \cref{eq:failed-extension} and the bound on the direct connections gives
\[
 P_t+\tau\Gamma_t>MS_{a_t},
\]
contradicting \cref{eq:allocation-condition}. Thus, for every $t$, all service events for clients in $A_t$ are assigned by the time vehicle $a_t$ has been processed. In particular, all service events are assigned.

Consider now the route constructed for a used vehicle $j$. By \cref{eq:segment}, its assigned subwalk together with the two direct connections to the depot forms a closed walk of length at most $Ms_j$. An unused vehicle receives the empty tour.

For HetTSP, shortcutting this closed walk gives a single tour starting and ending at $r$ and visiting exactly the clients assigned to vehicle $j$, without increasing its length. This tour forms the vehicle's route.

For both versions of HetCVRP, all intermediate visits to the depot are retained. Each nonempty tour in the resulting route is obtained from a subwalk of one of the capacity-feasible tours used to form $W$, with direct connections to the depot added at its ends when needed. It therefore retains only service events of that original tour, with unchanged delivered quantities, and its total delivery is at most $Q$. Hence every such tour is capacity-feasible. The tours, in their order, form the vehicle's route, whose total length is at most $Ms_j$.

Every service event is assigned exactly once and is never divided. In unsplit-delivery HetCVRP, the single service event of each client delivers its entire demand. In split-delivery HetCVRP, the delivered quantities fixed by the tour-splitting construction still sum to the client's demand, even if its service events are assigned to different vehicles. Thus the constructed solution is feasible for the corresponding problem. Since every vehicle $j$ has a route of length at most $Ms_j$, the constructed solution has makespan at most $M$.
\end{proof}

Thus every constructed solution considered by \cref{alg:main} is feasible, whether or not $\tau$ is itself a feasible makespan. The partition into vehicle groups affects the approximation guarantee, not feasibility; the remaining analysis bounds $P_t+\tau\Gamma_t$.

\subsection{Bounds from a fixed feasible solution}\label{sec:solution-bounds}
For the length analysis, fix a feasible HetTSP solution in the HetTSP case, and a feasible split-delivery HetCVRP solution for either HetCVRP algorithm. Let $B$ denote its makespan. For the final approximation guarantees, we take the fixed feasible solution to be optimal. Thus $B=\OPT_{\mathrm{HetTSP}}$ for HetTSP and $B=\OPT_{\mathrm{SD}}$ for both HetCVRP algorithms. Recall that $\OPT_{\mathrm{SD}}$ is a lower bound for the optimal value of unsplit-delivery HetCVRP.

Choose $\tau$ to be the largest member of $\Theta$ not exceeding $B$. Such a value exists because every client $v$ is visited by some vehicle $j$ with $2d(r,v)/s_j\le B$. The set $\Theta$ in \cref{eq:critical-times} contains every value at which a client's slowest reachable vehicle can change. Since no member of $\Theta$ lies in $(\tau,B]$, every client has the same slowest reachable vehicle at $\tau$ and at $B$. Hence $\tau$ is retained and
\begin{equation}\label{eq:shared-reachability}
 0<\tau\le B,\qquad e_{\tau}(v)=e_B(v)\quad(v\in V^\circ).
\end{equation}
This choice is used only in the analysis; the algorithm considers it among all retained tentative makespans in $\Theta$. Under \cref{eq:shared-reachability}, we use the routes of the fixed feasible solution to derive bounds for each layer client set. These bounds will control the length of $T_t$ and, for HetCVRP, the cumulative additional length introduced by tour splitting. They hold for any partition into vehicle groups as in \cref{eq:vehicle-groups}.

\begin{lemma}[Bounds for layer client sets]\label{lem:resources}
Under \cref{eq:shared-reachability}, the routes of the fixed feasible solution yield tours starting and ending at $r$, at most one for each vehicle $j\le b_t$, whose client sets are pairwise disjoint and have union $A_t$. The tour obtained from vehicle $j$ has length at most $Bs_j$. In particular,
\begin{equation}\label{eq:resource-bounds}
 \TSP(A_t)\le BS_{b_t}.
\end{equation}
If the fixed feasible solution is a split-delivery or unsplit-delivery HetCVRP solution, we also have
\begin{equation}\label{eq:cumulative-demand}
 \sum_{j=1}^t R_Q(A_j)\le BS_{b_t}\qquad(t=1,\ldots,m).
\end{equation}
\end{lemma}
\begin{proof}
Any route that visits client $v$ has length at least its round-trip length $2d(r,v)$. If vehicle $j$ visits $v$ and finishes within time $B$, then $2d(r,v)\le Bs_j$, so $j\le e_B(v)$. By \cref{eq:shared-reachability}, only vehicles $1,\ldots,b_t$ can visit clients in $A_t$ in the fixed feasible solution. Assign each client in $A_t$ to one vehicle that visits it in this solution. For each vehicle receiving clients, concatenate the tours of its route and shortcut the resulting closed walk to a tour starting and ending at $r$ and visiting exactly those clients. These tours have pairwise disjoint client sets with union $A_t$, and the tour obtained from vehicle $j$ has length at most $Bs_j$. Concatenating them at $r$ and shortcutting gives a tour starting and ending at $r$ and visiting exactly the clients in $A_t$, of length at most $BS_{b_t}$. This proves \cref{eq:resource-bounds}.

For \cref{eq:cumulative-demand}, let $X=A_1\cup\cdots\cup A_t$. The slowest reachable vehicle of every client in $X$ at $\tau$ is one of vehicles $1,\ldots,b_t$; by \cref{eq:shared-reachability}, the same is true at $B$. Hence all demand in $X$ is served by vehicles $1,\ldots,b_t$ in the fixed feasible solution. Consider one tour of one of these vehicles in the fixed feasible solution, of length $L$, and let $\xi_v$ be the quantity it delivers to $v\in X$, with $\xi_v=0$ if no delivery is made to $v$. Since $\sum_{v\in X}\xi_v\le Q$ and $d(r,v)\le L/2$ whenever $\xi_v>0$,
\[
 \frac2Q\sum_{v\in X}\xi_vd(r,v)\le\frac2Q\left(\sum_{v\in X}\xi_v\right)\frac L2\le L.
\]
Sum this inequality over the tours in the routes of vehicles $1,\ldots,b_t$ in the fixed feasible solution. Their total length is at most $BS_{b_t}$, and their deliveries satisfy all demands in $X$. Since the layer client sets are pairwise disjoint, $R_Q(X)=\sum_{j=1}^tR_Q(A_j)$, which proves \cref{eq:cumulative-demand}.
\end{proof}

\subsection{The tour bound and the depot saving}\label{sec:tour-analysis}
\begin{lemma}[Tour length for a layer client set]\label{lem:tour}
Under \cref{eq:shared-reachability}, so that $\tau\le B$ and every client has the same slowest reachable vehicle at $\tau$ and at $B$, the tour constructed on $A_t\cup\{r\}$ satisfies
\begin{equation}\label{eq:tour-bound}
 \len(T_t)\le\frac32BS_{b_t}-\frac{\tau}{2}b_t s_{b_t+1}.
\end{equation}
\end{lemma}
\begin{proof}
Put $h_t=\tau s_{b_t+1}/2$. For a nonempty $A_t$, take the tours from \cref{lem:resources} that are derived from the fixed feasible solution. Let $\mathcal I\subseteq\{1,\ldots,b_t\}$ be the indices of their nonempty tours, and denote those tours by $O_j$. Delete one traversal of an edge incident to the depot from every $O_j$, including one of the two traversals when $O_j$ contains a single client. By the lower bound in \cref{eq:layer-radius}, each deleted traversal has length at least $h_t$, so the path left from $O_j$ has length at most $\len(O_j)-h_t\le Bs_j-h_t$. These paths meet only at $r$ and together form a spanning tree on $A_t\cup\{r\}$. By the minimum spanning tree property of $F_t$,
\begin{equation}\label{eq:tree-bound}
 \len(F_t)\le\sum_{j\in\mathcal I}(Bs_j-h_t)
 \le\sum_{j=1}^{b_t}(Bs_j-h_t)
 =BS_{b_t}-b_th_t.
\end{equation}
The second inequality adds only nonnegative terms, since $Bs_j\ge\tau s_{b_t+1}=2h_t$ for every $j\le b_t$.

To bound the matching, list the vertices of $\odd(F_t)$ in their cyclic order along a tour attaining $\TSP(A_t)$, which starts and ends at $r$. The two alternating pairings have combined length at most $\TSP(A_t)$ by the triangle inequality. Consequently, the minimum-length perfect matching satisfies $\len(J_t)\le\TSP(A_t)/2\le BS_{b_t}/2$. The same statement holds when the odd set is empty. Adding this estimate to \cref{eq:tree-bound} and shortcutting proves \cref{eq:tour-bound}. If $A_t$ is empty, its tour has length zero, and the right side of \cref{eq:tour-bound} is nonnegative because $b_ts_{b_t+1}\le S_{b_t}$ and $\tau\le B$.
\end{proof}

The depot saving is $\frac{\tau}{2}b_ts_{b_t+1}$; \cref{sec:accumulated-analysis} uses it to bound the quantities $P_t+\tau\Gamma_t$ that determine $M$ in \cref{eq:allocation-makespan}.

\section{Vehicle groups and approximation guarantees}\label{sec:analysis-groups}
Fix $\lambda\in\{0,1,2\}$ as defined in \cref{sec:algorithm}. For the fixed feasible solution and retained tentative makespan $\tau$ from \cref{sec:solution-bounds}, we now show that some partition into vehicle groups in $\mathcal G_{K_\lambda}$ gives a value of $M$ at most the corresponding approximation factor stated in \cref{sec:results} times $B$.

With $B$ and $\tau$ chosen in \cref{sec:solution-bounds}, define $\beta=\tau/B\in(0,1]$ for the analysis. For any $K>2$, consider the partition into vehicle groups produced by the construction in \cref{sec:vehicle-groups} with $\eta=\beta$. By the definition of $\mathcal E_K$ and the fact proved in \cref{sec:vehicle-groups} that the resulting partition is unchanged between consecutive breakpoints, some value in $\mathcal E_K$ produces the same partition. Hence this partition belongs to $\mathcal G_K$.

\subsection{Accumulated tour lengths and direct connections}\label{sec:accumulated-analysis}
For the remainder of the analysis, fix $K>2$ and use the partition constructed in \cref{sec:vehicle-groups} with $\eta=\beta$, indexed from fastest to slowest as $[a_1,b_1],\ldots,[a_m,b_m]$. By the definition of $M$ in \cref{eq:allocation-makespan} and $\tau\le B$, proving $M\le\rho B$ for a factor $\rho\ge1$ reduces to bounding $P_t+\tau\Gamma_t$ by $\rho BS_{a_t}$ for every $t$.

In this numbering and with $\eta=\beta$, the rule in \cref{eq:partition-rule} says that each $a_t$ is the smallest index $j\in\{1,\ldots,b_t\}$ satisfying
\begin{equation}\label{eq:rule-final}
 \left(K-\frac32\right)S_{b_t}+\beta\left(j-\frac{b_t}{2}\right)s_{b_t+1}\le KS_j.
\end{equation}

To write the tour bound from \cref{lem:tour} more briefly, define
\begin{equation}\label{eq:beta-f}
 f_i(\beta)=\frac32S_i-\frac\beta2 i s_{i+1}\qquad(1\le i\le k).
\end{equation}
Then $\len(T_t)\le Bf_{b_t}(\beta)$ because $\tau=\beta B$. For HetTSP, $C_t=T_t$. For either version of HetCVRP, \cref{lem:partition} gives $\len(C_t)\le Bf_{b_t}(\beta)+\lambda R_Q(A_t)$.

We separate the tour bounds for the earlier layer client sets from the bound for the current tour by defining
\begin{equation}\label{eq:H}
 H_t=\sum_{j<t}f_{b_j}(\beta)+\beta\Gamma_t.
\end{equation}
Thus $BH_t$ bounds the total length of $T_1,\ldots,T_{t-1}$ plus $\tau\Gamma_t$. For HetCVRP, \cref{eq:cumulative-demand} bounds the additional length from tour splitting over all of $A_1,\ldots,A_t$ by $\lambda BS_{b_t}$. Summing the bounds on $C_1,\ldots,C_t$ therefore gives, for all three problems,
\begin{equation}\label{eq:combined-length}
 P_t+\tau\Gamma_t\le B\bigl(H_t+f_{b_t}(\beta)+\lambda S_{b_t}\bigr).
\end{equation}
The choice of the first vehicle in each group gives two useful inequalities. Since $a_t$ is the smallest index satisfying \cref{eq:rule-final}, the index $a_t-1$ fails it whenever $a_t>1$; this will bound the earlier tours and their contributions to $\beta\Gamma_t$ by the speeds within their respective groups. The inequality \cref{eq:rule-final} at $j=a_t$ then bounds the current tour after its depot saving. For HetTSP with $K=3$, these two inequalities bound the two contributions described in the introduction and give the factor $6$ directly. The next subsection gives the additional estimates needed for the general bound, including the two HetCVRP guarantees.

\Needspace{6\baselineskip}
\begin{lemma}[Accumulated length estimate]\label{lem:accumulated}
The partition constructed in \cref{sec:vehicle-groups} with $\eta=\beta$ satisfies $H_t\le KS_{a_t-1}$ for every $t=1,\ldots,m$. Moreover, if $\lambda=0$ and $K=3$, then $P_t+\tau\Gamma_t\le6BS_{a_t}$ for every $t=1,\ldots,m$.
\end{lemma}
\begin{proof}
Fix a vehicle group $[a,b]$. If $a>1$, the minimality of $a$ means that $j=a-1$ does not satisfy \cref{eq:rule-final}. Consequently,
\[
 KS_{a-1}\le\left(K-\frac32\right)S_b+\beta\left(a-1-\frac b2\right)s_{b+1}.
\]
When $a=1$, this inequality still holds: $S_0=0$, while the right side is at least $(K-2)S_b>0$ because $\beta bs_{b+1}\le S_b$. Rearranging and using \cref{eq:beta-f} gives
\begin{equation}\label{eq:group-length}
 f_b(\beta)+\beta(b-a+1)s_{b+1}\le K(S_b-S_{a-1}).
\end{equation}
Here $f_b(\beta)$ is the tour-length bound divided by $B$. For an earlier group, the second term is its contribution to $\beta\Gamma_t$, as the following identity shows. Since $a_{j+1}=b_j+1$,
\[
 \Gamma_t=\sum_{j<t}(a_{j+1}-a_j)s_{a_{j+1}}
 =\sum_{j<t}(b_j-a_j+1)s_{b_j+1}.
\]
Thus each earlier group contributes at most $K$ times the sum of the speeds within that group. The earlier groups partition the vehicles $1,\ldots,a_t-1$, so
\[
 \begin{aligned}
 H_t
 &=\sum_{j<t}\bigl(f_{b_j}(\beta)+\beta(b_j-a_j+1)s_{b_j+1}\bigr)\\
 &\le K\sum_{j<t}(S_{b_j}-S_{a_j-1})
 =KS_{a_t-1}.
 \end{aligned}
\]
This also covers $t=1$, when the sums are empty.

For HetTSP, let $\lambda=0$ and $K=3$. We now use that $a_t$ itself satisfies \cref{eq:rule-final}. At $K=3$, this gives
\[
 f_{b_t}(\beta)+\beta a_t s_{b_t+1}\le3S_{a_t},
\]
so $f_{b_t}(\beta)\le3S_{a_t}$. Combining the bounds for the earlier groups and the current tour in \cref{eq:combined-length} yields
\[
 P_t+\tau\Gamma_t\le B\bigl(H_t+f_{b_t}(\beta)\bigr)
 \le3BS_{a_t-1}+3BS_{a_t}\le6BS_{a_t}.
\]
\end{proof}

\subsection{Bounding the remaining term}
To derive the general bound, including the HetCVRP guarantees, we bound $f_{b_t}(\beta)+\lambda S_{b_t}$ in \cref{eq:combined-length} relative to $S_{a_t}$. For this purpose, we use \cref{eq:rule-final} at $j=a_t$:
\begin{equation}\label{eq:partition-right}
 \left(K-\frac32\right)S_{b_t}+\beta\left(a_t-\frac{b_t}{2}\right)s_{b_t+1}\le KS_{a_t}.
\end{equation}

The first bound below follows by dropping a nonnegative term from \cref{eq:partition-right}. The second bounds the total speed through the last vehicle of the group in terms of the total speed through its first vehicle.

\begin{lemma}[Consequences of the partition rule]\label{lem:local}
For every vehicle group $[a_t,b_t]$ in the partition constructed in \cref{sec:vehicle-groups} with $\eta=\beta$, we have
\begin{align}
 \left(K-\frac32\right)S_{b_t}-\frac\beta2 b_t s_{b_t+1}&\le KS_{a_t},\label{eq:local-one}\\
 S_{b_t}&\le\frac{K-\frac12}{K-2}S_{a_t}.\label{eq:local-x}
\end{align}
\end{lemma}
\begin{proof}
Dropping the nonnegative term $\beta a_t s_{b_t+1}$ from \cref{eq:partition-right} proves \cref{eq:local-one}.

For the second bound, the ordering of the vehicle speeds and $\beta\le1$ give $\beta(b_t-a_t)s_{b_t+1}\le S_{b_t}-S_{a_t}$. Therefore
\[
 \beta a_t s_{b_t+1}\ge\beta b_t s_{b_t+1}-(S_{b_t}-S_{a_t}).
\]
Substituting this lower bound into \cref{eq:partition-right} yields
\[
 \left(K-\frac52\right)S_{b_t}+\frac\beta2 b_t s_{b_t+1}
 \le(K-1)S_{a_t}.
\]
Adding this inequality to \cref{eq:local-one} cancels the terms involving $s_{b_t+1}$ and gives $(2K-4)S_{b_t}\le(2K-1)S_{a_t}$. Since $K>2$, division by $2K-4$ proves \cref{eq:local-x}.
\end{proof}

\subsection{Approximation factors}\label{sec:factors}
\begin{theorem}[Combined length bound]\label{thm:length}
Let $\lambda\in\{0,1,2\}$ and $K>2$. Fix a feasible HetTSP solution if $\lambda=0$, or a feasible split-delivery HetCVRP solution if $\lambda>0$, and denote its makespan by $B$. Let $\tau$ satisfy \cref{eq:shared-reachability}, so that $\tau\le B$ and every client has the same slowest reachable vehicle at $\tau$ and at $B$, and use the partition constructed in \cref{sec:vehicle-groups} with $\eta=\beta=\tau/B$. The corresponding closed walks $C_t$ from \cref{sec:allocation} satisfy
\begin{equation}\label{eq:final-length}
 P_t+\tau\Gamma_t\le B\rho_\lambda(K)S_{a_t}\qquad(t=1,\ldots,m),
\end{equation}
where
\begin{equation}\label{eq:rho}
 \rho_\lambda(K)=
 \begin{cases}
 K+\frac32+\lambda+\dfrac{\frac32(1+\lambda)}{K-2},&2<K\le\lambda+3,\\[2mm]
 K+\lambda+3,&K>\lambda+3.
 \end{cases}
\end{equation}
\end{theorem}
\begin{proof}
Fix $t$. By \cref{eq:combined-length,lem:accumulated} and $S_{a_t-1}\le S_{a_t}$,
\[
 \begin{aligned}
 P_t+\tau\Gamma_t
 &\le B\left(KS_{a_t}+\left(\frac32+\lambda\right)S_{b_t}-\frac\beta2 b_t s_{b_t+1}\right)\\
 &\le B\bigl(2KS_{a_t}+(\lambda+3-K)S_{b_t}\bigr),
 \end{aligned}
\]
where the second inequality follows from \cref{eq:local-one}.

If $2<K\le\lambda+3$, the coefficient $\lambda+3-K$ is nonnegative, so \cref{eq:local-x} gives
\[
 \begin{aligned}
 P_t+\tau\Gamma_t
 &\le B\left(2K+(\lambda+3-K)\frac{K-\frac12}{K-2}\right)S_{a_t}\\
 &=B\left(K+\frac32+\lambda+\frac{\frac32(1+\lambda)}{K-2}\right)S_{a_t}.
 \end{aligned}
\]
If $K>\lambda+3$, the coefficient is negative. Using $S_{b_t}\ge S_{a_t}$ instead gives $P_t+\tau\Gamma_t\le B(K+\lambda+3)S_{a_t}$. This proves both branches.
\end{proof}

For the general bound, the second branch of \cref{eq:rho} has infimum $2\lambda+6$, attained by the first branch at $K=\lambda+3$. It therefore suffices to minimize the first branch. Writing $u=K-2$, it equals $\frac72+\lambda+u+\frac{3(1+\lambda)}{2u}$ for $0<u\le\lambda+1$. Its unconstrained minimizer is $u=\sqrt{3(1+\lambda)/2}$. For $\lambda=0$ the minimum on this interval is at $u=1$; for $\lambda=1,2$ the unconstrained minimizer lies in the interval. This proves that the constants $K_\lambda$ in \cref{eq:K-values} give
\begin{equation}\label{eq:ratios}
 \rho_0(K_0)=6,\qquad \rho_1(K_1)=\frac92+2\sqrt3,\qquad \rho_2(K_2)=\frac{11}{2}+3\sqrt2.
\end{equation}

\subsection{Proofs of the main results}\label{sec:main-proofs}
We combine \cref{lem:allocation,thm:length} to bound the makespan returned by the algorithm relative to a feasible HetTSP solution or a feasible split-delivery HetCVRP solution, as appropriate.

\begin{lemma}[Guarantee relative to a feasible solution]\label{lem:algorithm-guarantee}
Run \cref{alg:main} with parameter $\lambda\in\{0,1,2\}$ and $K=K_\lambda$. Let $B$ be the makespan of a fixed feasible HetTSP solution if $\lambda=0$, or of a fixed feasible split-delivery HetCVRP solution if $\lambda>0$. The algorithm returns one route for each vehicle; together these routes form the output solution. The output solution is feasible for HetTSP if $\lambda=0$, split-delivery HetCVRP if $\lambda=1$, or unsplit-delivery HetCVRP if $\lambda=2$. Its makespan satisfies $\ALG\le\rho_\lambda(K_\lambda)B$.
\end{lemma}
\begin{proof}
For an empty client set, $\ALG=0$, so the assertion holds. Otherwise, every constructed solution is feasible by \cref{lem:allocation}. For the fixed feasible solution in the lemma, choose $\tau$ as in \cref{sec:solution-bounds}; it is retained and satisfies \cref{eq:shared-reachability}. As shown at the beginning of \cref{sec:analysis-groups}, $\mathcal G_{K_\lambda}$ contains the partition obtained by the construction in \cref{sec:vehicle-groups} with $\eta=\beta=\tau/B$. For these choices, \cref{thm:length} gives $P_t+\tau\Gamma_t\le B\rho_\lambda(K_\lambda)S_{a_t}$ for every $t$. Since $\tau\le B$ and $\rho_\lambda(K_\lambda)\ge1$, the definition of $M$ in \cref{eq:allocation-makespan} gives $M\le B\rho_\lambda(K_\lambda)$. The constructed solution corresponding to these choices has makespan at most $M$ by \cref{lem:allocation}. Since the output solution has no greater makespan, $\ALG\le M\le\rho_\lambda(K_\lambda)B$.
\end{proof}

The running-time theorem in \cref{sec:running-time} applies to all three algorithms. We use \cref{lem:algorithm-guarantee} to finish the three approximation proofs and the stronger guarantee following \cref{thm:ud}.

\begin{proof}[Proof of \cref{thm:hettsp}]
Use the HetTSP construction with $\lambda=0$ and $K=K_0=3$, and take $B=\OPT_{\mathrm{HetTSP}}$. By \cref{lem:algorithm-guarantee}, the output solution consists of one tour for each vehicle and satisfies $\ALG\le\rho_0(K_0)\OPT_{\mathrm{HetTSP}}=6\OPT_{\mathrm{HetTSP}}$.
\end{proof}

\begin{proof}[Proof of \cref{thm:sd}]
Use the split-delivery HetCVRP construction with $\lambda=1$ and $K=K_1=2+\sqrt3$, and take $B=\OPT_{\mathrm{SD}}$. By \cref{lem:algorithm-guarantee}, the output solution is feasible for split-delivery HetCVRP and satisfies $\ALG\le\rho_1(K_1)\OPT_{\mathrm{SD}}=(\frac92+2\sqrt3)\OPT_{\mathrm{SD}}$.
\end{proof}

\begin{proof}[Proof of \cref{thm:ud}]
Use the unsplit-delivery HetCVRP construction with $\lambda=2$ and $K=K_2=2+3/\sqrt2$. The tour-splitting construction and allocation deliver each client's entire demand in one visit, so the output solution is feasible for unsplit-delivery HetCVRP. Apply \cref{lem:algorithm-guarantee} with $B=\OPT_{\mathrm{SD}}$. It follows that
\[
\ALG\le\rho_2(K_2)\OPT_{\mathrm{SD}}=\left(\frac{11}{2}+3\sqrt2\right)\OPT_{\mathrm{SD}}\le\left(\frac{11}{2}+3\sqrt2\right)\OPT_{\mathrm{UD}},
\]
where the equality is \cref{eq:ratios} and the final inequality follows from $\OPT_{\mathrm{SD}}\le\OPT_{\mathrm{UD}}$. This proves the claimed approximation factor and the stronger bound stated after the theorem.
\end{proof}

\section{Running time}\label{sec:running-time}
We analyze \cref{alg:main} with $K=K_\lambda$, considering every retained tentative makespan and every partition into vehicle groups in $\mathcal G_K$.

\begin{theorem}[Running time]\label{thm:running-time}
Let $n=|V^\circ|$ and let $k$ be the number of vehicles. Each of the algorithms in \cref{thm:hettsp,thm:sd,thm:ud} can be implemented in
\begin{equation}\label{eq:running-time}
 O\!\left(n^4k^3+nk^4\right)
\end{equation}
time.
\end{theorem}
\begin{proof}
There are at most $nk$ tentative makespans. The set $\mathcal E_K$ contains $O(k^2)$ values. For each $\eta\in\mathcal E_K$, the construction in \cref{sec:vehicle-groups} forms at most $k$ vehicle groups and examines at most $k$ candidate indices in \cref{eq:partition-rule} for each group. Hence all partitions in $\mathcal G_K$ can be generated in $O(k^4)$ time and reused for every tentative makespan. For each tentative makespan, the slowest reachable vehicle of every client can be determined in $O(nk)$ time.

Fix a tentative makespan and a partition into vehicle groups, and let $p_t=|A_t|$. The layer client sets can be formed in $O(n+k)$ time, and $\sum_t p_t=n$. For a nonempty layer client set with $p$ clients, a minimum spanning tree in the complete graph can be computed in $O(p^2)$ time, while a minimum-length perfect matching on at most $p+1$ vertices can be found in $O(p^3)$ time \citep{Gabow2018}. The remaining operations needed to construct $T_t$ are linear in $p$. Therefore all tours for the fixed partition can be constructed in
\[
 O\!\left(\sum_t p_t^3\right)=O(n^3)
\]
time.

For either version of HetCVRP, the tour-splitting construction in Appendix~\ref{app:tour-splitting} takes $O(\sum_t p_t^2)=O(n^2)$ time. The total number of service events and edge traversals in the resulting closed walks is $O(n)$. Computing $P_t$, $\Gamma_t$, and $M$, performing the ordered allocation, and constructing the vehicle routes take $O(n+k)$ additional time.

Thus one tentative makespan and one partition into vehicle groups require $O(n^3+k)$ time. Since there are $O(nk)$ tentative makespans and $O(k^2)$ partitions in $\mathcal G_K$, there are $O(nk^3)$ such pairs. Their total running time is
\[
 O(nk^3(n^3+k))
 =
 O\!\left(n^4k^3+nk^4\right).
\]
The preprocessing time is dominated by this bound.
\end{proof}

\section{Conclusion}\label{sec:conclusion}
We gave polynomial-time approximation algorithms with factors $6$ for HetTSP, $\frac92+2\sqrt3$ for split-delivery HetCVRP, and $\frac{11}{2}+3\sqrt2$ for unsplit-delivery HetCVRP. The unsplit-delivery guarantee also holds relative to $\OPT_{\mathrm{SD}}$.

\paragraph{Use of generative AI.}
The author used OpenAI's ChatGPT and Anthropic's Claude in the preparation of this manuscript. The algorithmic framework was developed by the author. In an earlier version of the algorithm and its analysis, the author constructed the tour for each layer client set by applying an $\alpha$-approximation algorithm for TSP using only its approximation guarantee, together with a different choice of parameters for the vehicle groups, and obtained a $(4\alpha+1)$-approximation for HetTSP. In a conversation with ChatGPT, ChatGPT suggested replacing this step by a minimum-weight spanning tree with a degree bound at the depot, using modified weights on the edges incident to the depot, followed by a minimum-length perfect matching on its odd-degree vertices, and adjusting the parameters accordingly. With this construction and the adjusted parameters, the analysis yields the depot saving and hence the $6$-approximation for HetTSP. The author then applied the same technique to obtain the approximation factors for split-delivery and unsplit-delivery HetCVRP. The author subsequently found that the degree bound and the modified weights are unnecessary: the same bound holds for a minimum spanning tree in the original distances, so each tour can be constructed by the Christofides--Serdjukov algorithm, as presented here. ChatGPT and Claude were also used for language editing, presentation, and checking proofs. The author has verified all results and proofs and takes full responsibility for the content of this manuscript.

\appendix
\section{Details of the classical tour-splitting constructions}\label{app:tour-splitting}
We recall the classical constructions underlying \cref{lem:partition}, following the presentation in \citet[Theorem~5]{BTV} of the results of \citet{HRK,AG87}. We include the details needed to specify the client visits and delivered quantities used later in \cref{sec:allocation}.

\paragraph{Common setup.}
Let the clients in $X$ occur along $T$ in the order $v_1,\ldots,v_N$, and write $\bar q_j=q_{v_j}/Q$, $\Delta_0=0$, and $\Delta_j=\sum_{h=1}^j\bar q_h$. Represent the demand of $v_j$ by the interval $(\Delta_{j-1},\Delta_j)$. Choose an offset $\theta$ uniformly in $[0,1)$ and consider the points of $\theta+\mathbb Z$, avoiding offsets for which a marked point coincides with an endpoint $\Delta_j$.

\paragraph{Unsplit-delivery version.}
Call a client selected if its demand interval contains a marked point. Serve each selected client on a tour containing only that client and the depot. Each maximal consecutive sequence of unselected clients is served by one tour following their order along $T$. Since the corresponding demand intervals lie between two consecutive marked points, each such tour is capacity-feasible.

Let $Z$ be the set of selected clients. If $u$ and $w$ are the last and first clients of two consecutive tours serving unselected clients, and $z\in Z$ lies between them along $T$, then the triangle inequality gives
\[
 d(u,r)+d(r,w)\le \len(T[u,w])+2d(r,z).
\]
It follows that the tours serving unselected clients have total length at most $\len(T)+2\sum_{v\in Z}d(r,v)$. The single-client tours for the selected clients contribute another $2\sum_{v\in Z}d(r,v)$. Since client $v_j$ is selected with probability $\bar q_j$,
\[
 \mathbb E[\text{total length}]
 \le
 \len(T)+4\sum_{j=1}^N\bar q_jd(r,v_j)
 =
 \len(T)+2R_Q(X).
\]
Each client is visited exactly once.

\paragraph{Split-delivery version.}
Use the marked points to divide $(0,\Delta_N)$ into intervals of length at most one. For each such interval, construct a tour visiting, in their order along $T$, the clients whose demand intervals intersect it, and deliver to each client $Q$ times the length of the intersection. Each resulting tour is capacity-feasible. If a marked point lies inside the demand interval of a client $v$, that client's demand is split between the two adjacent tours; in the concatenated walk this introduces a return to the depot and a departure back to $v$, of total length $2d(r,v)$. Hence
\[
 \mathbb E[\text{total length}]
 \le
 \len(T)+2\sum_{j=1}^N\bar q_jd(r,v_j)
 =
 \len(T)+R_Q(X).
\]
Since $\bar q_j\le1$, each client is visited by at most two resulting tours.

\paragraph{Deterministic implementation.}
The random offset is used only to prove the length bounds. As in \citet[Theorem~5]{BTV}, it suffices to consider one representative offset from each interval determined by the residues $\Delta_j\bmod1$ and retain a construction of minimum total length. There are $O(N)$ such choices, and each can be evaluated in $O(N)$ time, giving an $O(N^2)$ implementation.

The resulting tours can be ordered according to their occurrences along $T$ and concatenated at the depot. In the unsplit-delivery version each client then has one service event, while in the split-delivery version it has at most two, with delivered quantities summing to its demand. Thus the construction produces the closed walk and service events used in \cref{sec:allocation}, with $O(N)$ visits and edge traversals.

\end{document}